\documentclass[transaction]{IEEEtran}

\usepackage{amsfonts}
\usepackage{amssymb}
\usepackage{hyperref}
\usepackage{graphicx}
\usepackage{enumerate}
\usepackage{amsmath}

\newtheorem{theorem}{Theorem}
\newtheorem{proposition}{Proposition}

\begin{document}
%
\title{Non-stationary Resource Allocation Policies for Delay-constrained Video Streaming: Application to Video over Internet-of-Things-enabled Networks}
%
%
%

\author{Jie~Xu,~\IEEEmembership{Student Member,~IEEE},
        Yiannis Andrepoulos,~\IEEEmembership{Member,~IEEE},\\
        Yuanzhang~Xiao,~\IEEEmembership{Student Member,~IEEE},
        and~Mihaela~van~der~Schaar,~\IEEEmembership{Fellow,~IEEE}
\thanks{J. Xu, Y. Xiao and M. van der Schaar are with the Department of Electrical Engineering, University of California Los Angeles (UCLA), Los Angeles, CA, 90095, USA. Emails: jiexu@ucla.edu, yxiao@ee.ucla.edu, mihaela@ee.ucla.edu. Y. Andrepoulos is with Department of Electronic and Electrical Engineering, University College London (UCL), London, UK. Email: iandreop@ee.ucl.ac.uk. J. Xu, Y. Xiao and M. van der Schaar were supported by US NSF grant CNS-1016081. Y. Andrepoulos was partially supported by the UK EPSRC, EP/K033166/1.}}

\maketitle

\begin{abstract}
Due to the high bandwidth requirements and stringent delay constraints of multi-user wireless video transmission applications, ensuring that all video senders have sufficient transmission opportunities to use before their delay deadlines expire is a longstanding research problem. We propose a novel solution that addresses this problem without assuming detailed packet-level knowledge, which is unavailable at resource allocation time (i.e. prior to the actual compression and transmission). Instead, we translate the transmission delay deadlines of each sender's video packets into a monotonically-decreasing weight distribution within the considered time horizon. Higher weights are assigned to the slots that have higher probability for deadline-abiding delivery. Given the sets of weights of the senders' video streams, we propose the low-complexity Delay-Aware Resource Allocation (DARA) approach to compute the optimal slot allocation policy that maximizes the deadline-abiding delivery of all senders. A unique characteristic of the DARA approach is that it yields a non-stationary slot allocation policy that depends on the allocation of previous slots. This is in contrast with all existing slot allocation policies such as round-robin or rate-adaptive round-robin policies, which are stationary because the allocation of the current slot does not depend on the allocation of previous slots. We prove that the DARA approach is optimal for weight distributions that are exponentially decreasing in time. We further implement our framework for real-time video streaming in wireless personal area networks that are gaining significant traction within the new Internet-of-Things (IoT) paradigm. For multiple surveillance videos encoded with H.264/AVC and streamed via the 6tisch framework that simulates the IoT-oriented IEEE 802.15.4e TSCH medium access control, our solution is shown to be the only one that ensures all video bitstreams are delivered with acceptable quality in a deadline-abiding manner.
\end{abstract}
%
%

%
\IEEEpeerreviewmaketitle

\vspace{-15pt}
\section{Introduction}
Multi-user wireless resource allocation for multiple video bitstream transmitters is a longstanding research problem \cite{vdS2005}\cite{Huang}-\cite{vdS2006}. One of the major challenges in such communications systems is how to allocate transmission resources (i.e. timeslots) in a manner that allows all video senders to send sufficient amount of packets prior to their deadline expiration and thus ensure sufficient video quality at the receiver side. This problem is becoming particularly pertinent today with the emergence of the Internet-of-Things (IoT) and machine-to-machine (M2M) communications that are expected to bring a disruptive change in the way people access real-time sensor data flows across the globe \cite{Watteyne}\cite{Wang}. This is because new standards such as the IEEE 802.15.4e at the medium access control (MAC) layer \cite{IEEE 802.15.4} (as finalized in 2012) and 6LoWPAN at the network layer \cite{Shelby} now pave the way for a unified IPv6-based network layer between wireless visual sensors and arbitrary destinations across the Internet (see Figure 1). Within such a paradigm the challenging part for delay-constrained video transmission is the IEEE 802.15.4e-enabled wireless personal area network (WPAN). Such WPANs are based on a central coordinator that also serves as the gateway to the broader Internet and is thus called the low power border router (LPBR) \cite{IEEE 802.15.4}. The functionality of the LPBR is two-fold \cite{Watteyne}\cite{Wang}\cite{IEEE 802.15.4}: firstly to coordinate the timeslots provided to each sensor during the network active time and secondly to aggregate all received streams and forward them to the destination IPv6 addresses \cite{Watteyne}\cite{Wang}\cite{IEEE 802.15.4}. Efficient allocation of timeslots to each visual sensor is crucial to the success of real-time video streaming due to its high bandwidth requirements and stringent delay constraints. Moreover, the allocation should be done without detailed knowledge about the distortion impact and delay deadlines of each video packet, which tends to be unavailable or overly-complex to obtain in real time.

\begin{figure}
\centerline{\includegraphics[angle = 270, scale = 0.25]{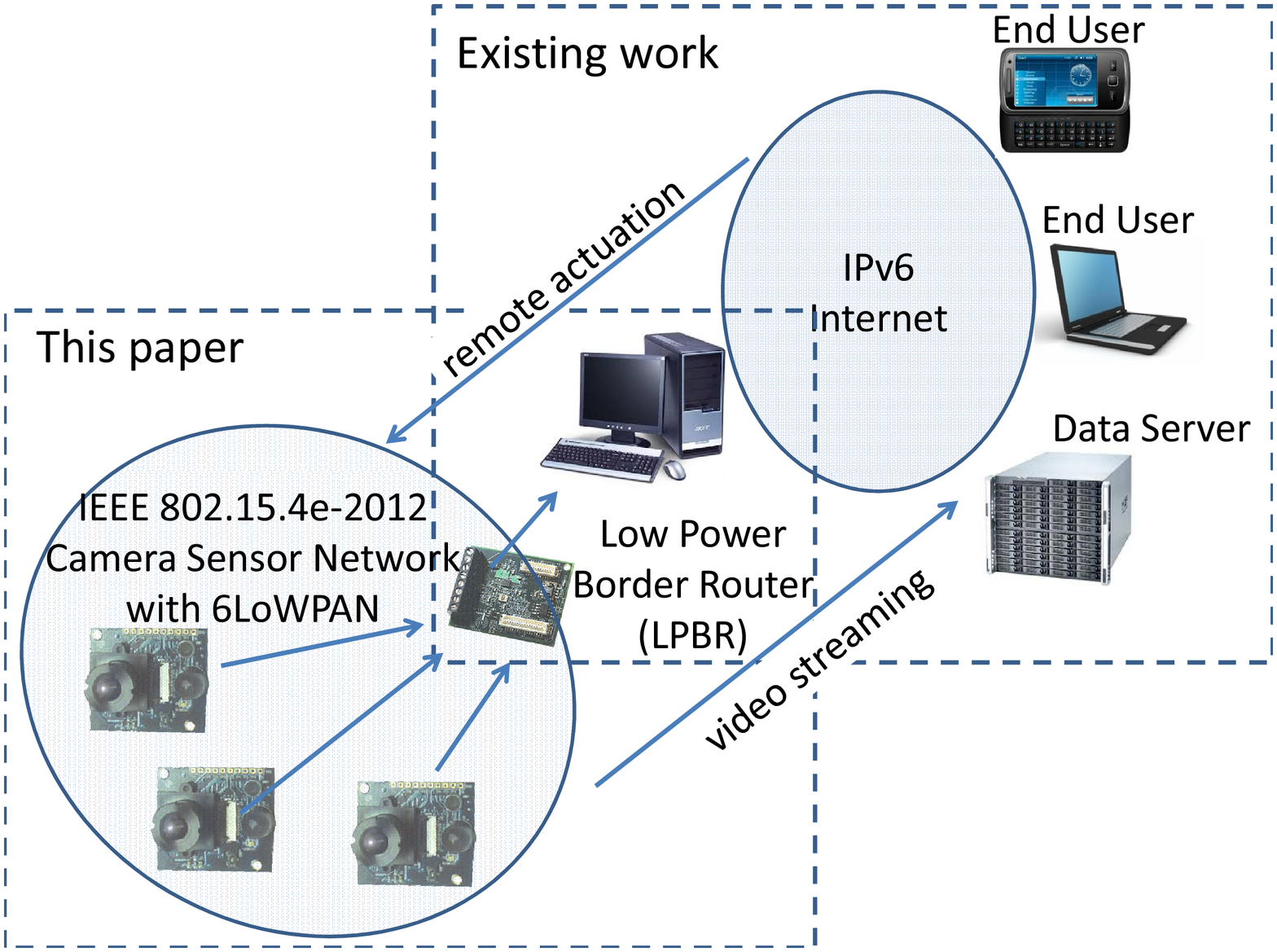}}
\caption{Video streaming over IoT-oriented standards.}\label{system}
\vspace{-20pt}
\end{figure}

In this paper, we propose for the first time an  optimal  multi-user resource allocation framework for delay-constrained video streaming, which is based solely on statistical information about the video users' characteristics, rigorously prove its optimality and show how it can be applied to multi-camera video resource allocation in IEEE 802.15.4e WPANs under the newly-standardized time synchronized channel hopping  (TSCH) MAC \cite{Watteyne}\cite{Wang}\cite{IEEE 802.15.4}, which allows for synchronization and contention-free and interference-mitigated transmission without the use of complex collision avoidance mechanisms \cite{IEEE 802.15.4}. Even though we validate our slot allocation solution for this particular setting, because it provides ways to coordinate users with delay preferences using only limited information, our proposal is applicable to many other resource allocation problems in many other settings, such as IEEE 802.11e WLANs, video streaming over 3G/LTE cellular networks, etc.

\vspace{-10pt}
\subsection{Problem Description}
The importance of timeslots (in expectation) is highly related to their position - the earlier the slot becomes available to a sensor, the more useful is it to its packet transmission since it provides more laxity for the deadline requirement. To illustrate this, Figure 2 presents deadline distributions calculated for four different video bitstreams within the same slotframe interval (which is set to one second). Within this interval, these distributions present the sizes of the video bitstream of each sensor with transmission deadline after the time marked in the horizontal axis\footnote{These results were generated from CIF-10Hz surveillance videos encoded with the H.264/AVC encoder under: high profile, low-delay IBBP encoding structure, transmission deadline set to 100 ms  and intra-frame period of 4 s.}. The upper-left graph (sensor 1) demonstrates a sharp peak at $t=0$  as it includes an intra (I) frame, while all other three graphs only include predicted (P) and bidirectionally predicted (B) frames within the specific slotframe interval. Evidently, under the strict transmission deadline of 100 ms, providing MAC-layer timeslots in a round-robin fashion to these four sensors will not be optimal as, for example, sensor 1 requires significantly more timeslots at the first 100 ms of the slotframe (to accommodate the I-frame deadline) than sensors 2$\sim$4 that have a much smoother deadline distribution. This indicates that using weights derived by such deadline distributions will automatically incorporate a measure of the expected quality of the received video bitstream corresponding to the slotframe interval under consideration.

\begin{figure}
\centerline{\includegraphics[angle = 270, scale = 0.35]{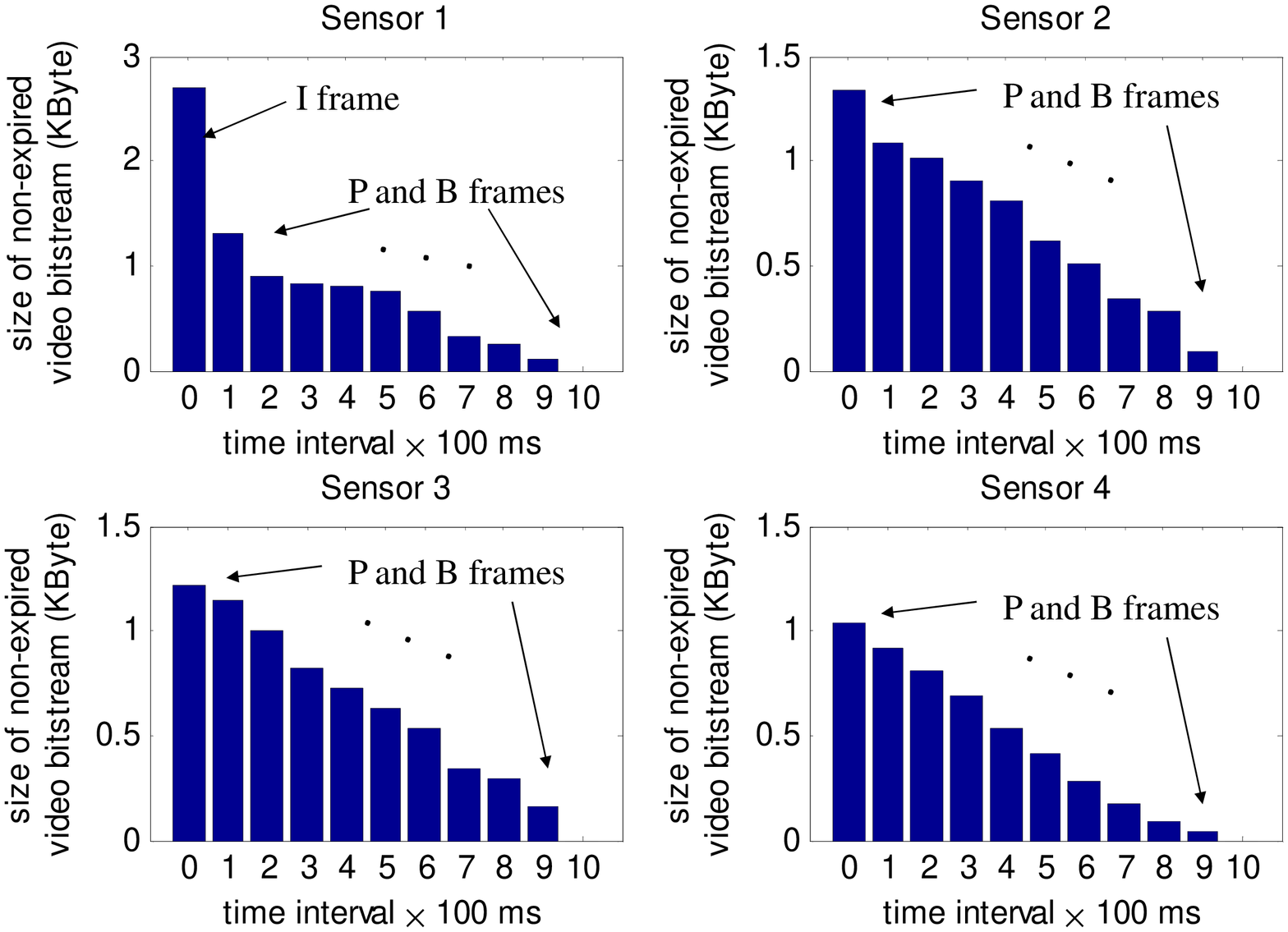}}
\caption{Distributions of four video bitstreams from four visual sensors presenting the size of each video bitstream within the same slotframe (1-second interval) that has transmission deadline after the marked time.}\label{histogram}
\vspace{-18pt}
\end{figure}

\vspace{-10pt}
\subsection{Contribution and Paper Organization}
In this paper we model the value of different MAC-layer timeslots to visual sensors by assigning weights to slots for each sensor, which represent the sensors' discounting of the value of future transmission opportunities. These weights depend on only some statistical information of the sensors' video content, video application requirements and video encoding/decoding techniques (e.g. the H.264/AVC temporal prediction structure) but not the specific packet-level information. Since earlier timeslots will satisfy sensors' delay requirements more easily, the only constraint is that weights are monotonically decreasing in time. Given the weights of each sensor, we propose an efficient approach to perform MAC-layer resource allocation for multi-camera WPANs, called {\it Delay-Aware Resource Allocation} (DARA).  The DARA approach assigns each sensor with a single index that captures three important aspects for its resource (i.e. timeslot) allocation: (i) the distance from the target timeslot allocation; (ii) the benefit of allocating a timeslot to a sensor; and (iii) the discounted sum of remaining transmission opportunities for a sensor in the same slotframe. It allocates the current timeslot to the sensor with the largest current index and then updates each sensor's slot indices. Hence, a sensor with: (i) larger distance from the target; (ii) a larger current benefit and (iii) fewer discounted remaining transmission opportunities, is more likely to be assigned with the current timeslot. For exponentially-decreasing weight distributions that were found to characterize well the observed deadline distributions of real video traces, we prove that our approach achieves the optimal performance. We also show via numerical experiments and real-world video streaming over the 6tisch simulator of the IEEE 802.15.4e TSCH MAC \cite{Watteyne}\cite{Wang} that our approach significantly outperforms existing solutions.

A unique characteristic of the proposed DARA approach is that it yields a non-stationary slot allocation policy. We define the notion of a (non-)stationary allocation below.

{\it Definition.} A deterministic or probabilistic slot allocation process is stationary if each slot assignment depends only on the available (finite) set of selection states and the user preferences and does not depend on time. Otherwise, a slot allocation is non-stationary.

{\it Examples}: Round-robin allocation is stationary and so is random slot allocation with fixed probability for each slot. A weighted round-robin allocation, with the weights depending on the user preferences (e.g. average bitrates and delay deadlines) but not on time, is also stationary. A random slot allocation with each slot having time-varying probability of being assigned to a video sender is non-stationary.

As we shall see, under an appropriate optimization framework, non-stationary allocations can better adapt to the different users' resource requirements over time and hence, yield a much better performance \cite{Xiao}. We discuss this in more detail in the following subsection.

\vspace{-10pt}
\subsection{Review of State-of-the-art in Video Streaming over Wireless Networks}
Existing wireless video streaming solutions that could potentially fit into IEEE 802.15.4e-enabled networks can be broadly divided into three categories. The first category encompasses single-user video streaming solutions, focusing on packet scheduling, error protection or cross-layer adaptation in order to maximize the received video quality \cite{vdS2005}-\cite{Liang}. Such proposals assume exact packet-level information (e.g. packet distortion impact and packet-level transmission deadline) and propose highly-complex rate-distortion optimized packet scheduling solutions. Since only single-user transmission scenarios are considered, such proposals are useful after the multi-user resource (i.e. slot) allocation problem is addressed and can easily interoperate with our proposed slot allocation solution.

The second category comprises multi-user video streaming, emphasizing on resource allocation amongst multiple users (i.e. visual sensors in our case) simultaneously transmitting video and sharing the same wireless resources \cite{Huang}-\cite{Vukadinovic}. The network utility maximization (NUM) framework and optimization-based algorithms have been widely studied for solving wireless multi-user resource allocation problems in the past few years \cite{Eryilmaz}-\cite{Joe-Wong}. However, these models do not take into account the different delay preferences of difference users over time in the video streaming problem. Hence, the resulting slot allocations are stationary under our definition. Detailed packet-level information of the video bitstream and channel state information are incorporated in the optimized resource allocation solution. In Huang et al \cite{Huang}, a NUM problem is first solved for multi-user rate allocation based on which time slots are assigned to users according to the deadlines of their packets. The optimal multi-user delay-constrained wireless video transmission framework is formulated as a multi-user Markov decision process (MUMDP) \cite{Fu2010}. The MUMDP then is decomposed into local MDPs which can be autonomously solved by individual users. Even though such proposals provide for joint throughput and delay based optimization, detailed knowledge of the video packet contents and their corresponding distortion impact is required to perform the resource allocation. However, such solutions are unsuitable within IoT-enabled networks of sensors, where it is extremely unlikely that: (i) either the sensors or the LPBR will have access to the individual packet distortion estimates and delay requirements; (ii) either system will have the resources necessary to process and derive an optimal transmission allocation based on such information.

The third category of research works proposes slot allocation policies for multi-user video transmission without packet-level knowledge and is instead based on knowledge of only some statistics, such as the histogram of the delay deadlines and/or the histogram of bitstream element sizes present within the resource allocation period of the video streams. Examples of such policies are the round-robin policy (e.g. weighted-round-robin or ``water-filling'' strategies). For instance, the sensors can be set to transmit in some predetermined order, but, within each slot allocation block, each sensor may be allocated a number of slots that is proportional to its rate or delay sensitivity. In Dutta et al  \cite{Dutta}, a greedy policy is proposed to fairly allocate transmission slots amongst multiple variable-bitrate streaming videos in order to maximize the minimum ``playout lead'' (i.e. duration of time the video can be played using only the data already buffered in its client) across all videos. When the video streams are compressed using constant-bitrate encoding, the resulting policy essentially reduces to a weighted round-robin policy with weights corresponding to the video bitrate. In Pradas and Vazquez-Castro \cite{Pradas}, a NUM-based framework is developed to balance rate-delay performance for video multicasting over adaptive satellite networks. Video streams are classified into, so-called, Classes of Services (CoS) with different delay requirements. The NUM framework then incorporates these requirements by solving a weighted sum utility maximization problem (where the weights are based on the delay requirements). After the video transmission rates are determined, a weighted-round-robin policy is performed to allocated slots to different video streams depending on their CoS. Though not exactly the same because of the different deployment under consideration, this policy is analogous to a weighted-round-robin policy is which the weights are determined based on both the rate and the delay requirements.

Because our proposal falls into this category of approaches that do not require packet-level information, and to better illustrate the differences with prior works, we provide an intuitive comparison of our resource allocation with two benchmark policies (which we call "R-Round-Robin" and "RD-Round-Robin") that encapsulate the merits of Dutta et al \cite{Dutta} and Pradas and Vazquez-Castro \cite{Pradas}, respectively. It is important to note that these policies are stationary because the allocation of the current slot does not depend on, and adapt to, the allocation of previous slots within each slot allocation block. As we shall see, we can significantly improve upon the performance of these stationary polices by using non-stationary policies that take into account the allocation of previous slots \cite{Xiao}.

To better illustrate the difference of these policies, Table I shows an example slot allocation for 3 sensors A, B and C within one slotframe comprising 12 timeslots. The target rate of sensor A is three times of that of sensor C. The target rate of sensor B is twice of that of sensor C. Sensor B is the most delay sensitive while sensor C is the least delay sensitive. The allocations of all three benchmark policies result in cyclic repetitions within the slotframe, while in the proposed DARA policy the allocation depends on past allocations and is non-cyclic. Table II further presents a summative comparison of our approach with the state-of-the-art on multi-user video streaming.

\begin{table}
\caption{An illustrative example of the resulting allocation of different policies within one slotframe.}\label{table1}
\vspace{-10pt}
\centerline{\includegraphics[scale = 0.78]{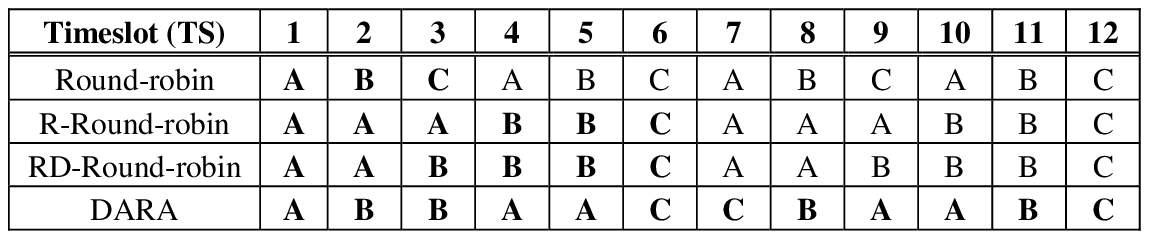}}
\vspace{-10pt}
\end{table}

\begin{table}
\caption{Comparison of different policies.}\label{table2}
\vspace{-10pt}
\centerline{\includegraphics[scale = 0.75]{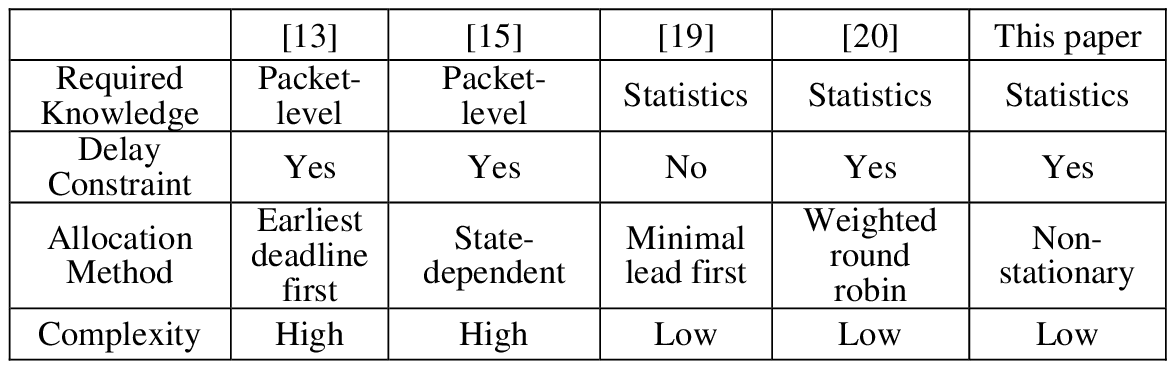}}
\vspace{-10pt}
\end{table}

The remainder of this paper is organized as follows: Section II describes the system model and formulates the resource allocation problem with incomplete information. Section III proposes the DARA approach for resource allocation. Section IV derives analytic bounds that characterize the performance of the DARA approach. Section V provides the numerical as well as experimental results using real video sequences on a real deployment. Section VI concludes the paper.

\vspace{-10pt}
\section{System Model}
We consider an LPBR under the TSCH mode of IEEE 802.15.4e MAC with $N$ visual sensors sharing the TSCH slotframe for their video bitstream transmission. Since we cannot control the end-to-end delay over the IoT scenario of Figure 1, we consider instead a deadline set for each video bitstream part produced by each visual sensor. This transmission deadline is in the order of 100 ms. As shown in Figure 1, all sensors are directly connected to the LPBR that relays their video bitstreams to the end users using well-established streaming over TCP/IP or UDP/IP \cite{Stockhammer}\cite{Tsakos}.  Table III provides a table of the key notations used in this paper.

\begin{table}
\caption{Nomenclature table.}\label{table3}
\vspace{-10pt}
\centerline{\includegraphics[scale = 0.8]{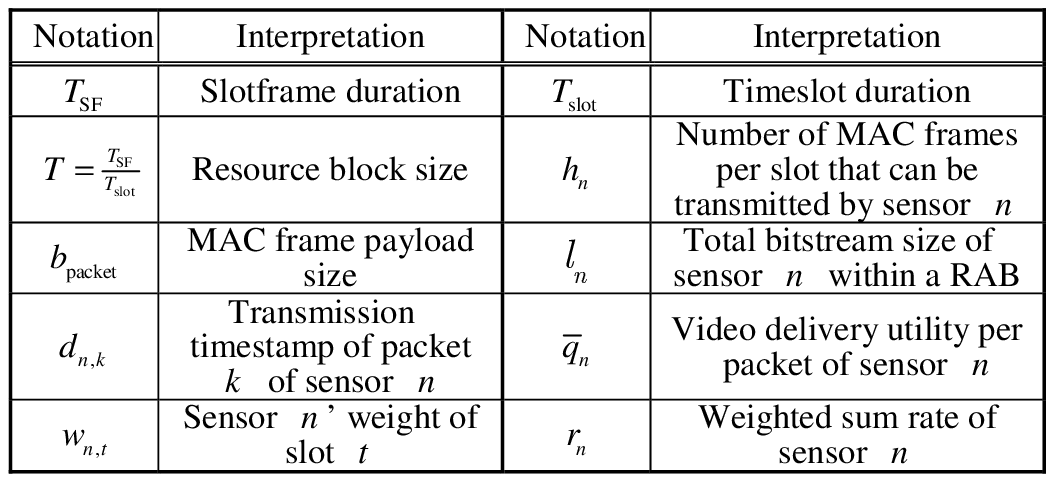}}
\vspace{-15pt}
\end{table}

\vspace{-10pt}
\subsection{Wireless System Abstraction}
Under the TSCH mode, within each slotframe interval of $T_{\textrm{SF}}$ s, each sensor is allocated one or more timeslots (and a corresponding channel) for each of its transmission opportunities. Thus, within each $T_{\textrm{SF}}$ s, each sensor will hop to one or more of the 16 channels available in the IEEE 802.15.4 PHY, albeit for very brief intervals of time \cite{Watteyne}\cite{Wang}\cite{IEEE 802.15.4}. Under TSCH, the transmission medium coherence time experienced at the MAC layer is substantially longer than $T_{\textrm{SF}}$ \cite{Watteyne}\cite{Wang}\cite{Bachir}\cite{Pister}. That is, the error rate at the MAC layer can be regarded as constant over any time interval smaller or equal to $T_{\textrm{SF}}$. In fact, due to TSCH's properties this assumption has been shown to be valid even for substantially prolonged periods of time if the sensors and LPBR are not moving rapidly \cite{Bachir}\cite{Pister}. Moreover, due to strong channel coding schemes employed at the IEEE 802.15.4 PHY, less than 1\% packet loss is observed at the MAC layer under typical operational conditions \cite{Pister}. Consequently, in the following we focus on the resource allocation performed in groups of $T=\frac{T_{\textrm{SF}}}{T_{\textrm{slot}}}$ timeslots, with $T_{\textrm{slot}}$  being the duration of one timeslot, and also assume very limited or no packet retransmissions at the MAC layer. These $T$ timeslots form a {\it resource allocation block} (RAB). The number of MAC frames that can be transmitted by sensor $n$, $1\leq n \leq N$, within each timeslot of a RAB is denoted by $h_n$ and we assume that each MAC frame carries a payload of $b_{\textrm{packet}}$ bytes. Typical values for these parameters in TSCH MAC are \cite{Watteyne}\cite{Wang}\cite{Bachir}\cite{Pister}: $T_{\textrm{SF}}\in [0.2, 2]$ s, $T_{\textrm{slot}} \in [6, 10]$ ms, $h_n \in \{1,2,3\}$, $b_{\textrm{packet}} \in [45, 110]$ bytes.

\vspace{-10pt}
\subsection{Video Coding Adaptivity and Transmission Abstraction}
During each RAB, each sensor $n$  needs to send the bitstream corresponding to several independently-decodable video bitstream units\footnote{i.e. bitstreams corresponding to $\{I, P, B\}$  frames that can be decoded if they are received in their entirety, with the application of concealment - if necessary.}. We denote the total bitstream size (bytes) of each sensor within one RAB by $l_n$. Any MPEG video coder can be used to generate the video bitstream based on any standard $\{I, P, B\}$ temporal prediction structure with a fixed transmission deadline set for each compressed  $\{I, P, B\}$ frame once it is produced by the encoder. This transmission deadline is imposed either by the limited on-board memory of each sensor (e.g. when only limited number of compressed frames can be buffered), or due to the stringent delay constraints imposed by the application context (e.g. in real-time video surveillance or monitoring). Each sensor adjusts its number of packets by discarding independently-decodable bitstream parts that have expired. For example, within an MP4 or MKV encapsulation of H.264/AVC video \cite{Tsakos}, it is straightforward to discard NAL units or SimpleBlocks (respectively) with expired transmission deadlines and a partially received bitstream can be reliably decoded with the FFmpeg library, which will also apply error concealment before displaying the decoded video.

Within each RAB, each packet $k\in \{1,...,\frac{l_n}{b_{\textrm{packet}}}\}$  of sensor $n$  bears a transmission time stamp (TTS) $d_{n,k}$  (measured in number of timeslots with respect to the beginning of the current RAB). The packet is valid for transmission if and only if it is sent by the sensor before (or by) slot $d_{n,k}$. Finally, each packet of sensor $n$  transmitted until slot $d_{n,k}$ induces an expected {\it video delivery utility} $\bar{q}_n$, which is a measure of the average improvement of video quality if this packet is used. We remark that, while distortion estimates can be used for $\bar{q}_n$ following previous work \cite{Chou}\cite{vdS2007}, in this paper we define utilities based on the expected deadline and expected video delivery utility of groups packets, as this is information that can be reliably (and quickly) estimated by the video encoder prior to the actual production of video packets.

\vspace{-10pt}
\subsection{Problem Formulation}
Since the LPBR has no {\it a-priori} knowledge of the characteristics of the generated video bitstream of each sensor (i.e. no knowledge of $\forall n, k: d_{n,k}$), it assigns the slots of each RAB to sensors solely based on limited statistical information about each video bitstream. This information corresponds to each sensor's weight distribution and the expected video delivery utility $\bar{q}_n$, which is transmitted from each sensor to the LPBR periodically with very low overhead (only two MAC layer packets are needed).

Let $\bold{s}$  be the current RAB vector, with elements $s(t)\in\{1,2,...,N\}$ representing the sensor to which timeslot $t\in \{1,..., T\}$ is allocated. If the resource allocation incorporates the entire bitstream of sensor $n$  without any transmission deadline violation, then the number of MAC frames transmitted is $h_n \sum\limits_{t=1}^T \textbf{1}(s(t) = 1)$, where $\textbf{1}(\cdot)$  is the indicator function. In such a case, we define the utility of the sensor $n$ given the RAB vector $\bold{s}$ by $Q_n(\bold{s}) = \bar{q}_n h_n \sum\limits_{t=1}^T \textbf{1}(s(t)=n)$ which gives a measure of the expected video quality  for sensor $n$  in the current RAB \cite{Chou}\cite{Liang}. Because in the actual deployment the resource allocation is performed before the video packets are generated, the utility is not considering the visual significance of each packet per-se but is computed based on the expected video delivery utility, $\bold{q}_n$, and the amount of MAC frames transmitted. However, it is possible that some packets are not able to meet their TTS deadline if the timeslots for sensor $n$  come too late. In a multi-camera sensor network, this problem becomes more complex since different sensors have different TTS deadlines for their packets. In this paper, we model the different delay sensitivities amongst the $N$  sensors by the sensors' discounting of the value of upcoming timeslots within the slotframe. Let $\bold{w}_n$  be the weight vector of size $T$  for sensor $n$, where $w_{n,t}\in [0,1]$ is $n^{\textrm{th}}$ sensor's weight of timeslot $t$. We normalize the weight vector by setting $w_{n,1} = 1$ and, by definition: $\forall n, t: w_{n,t} \geq w_{n, t+1}$, since the sensor's valuation of transmission opportunities is decreasing with time as packets begin to expire. Examples of $\bold{w}_n$ for four different sensors are given in Figure 2. The expected discounted utility of sensor $n$  is then given by:
\begin{align}
Q_n(\bold{s}, \bold{w}_n) = \bar{q}_n h_n \sum\limits_{t=1}^T w_{n,t} \textbf{1}(s(t) = n)\label{eq2}
\end{align}
The goal of the LPBR is to maximize an objective function of the sensors' discounted utilities  $W(Q_1(\bold{s}, \bold{w}_1), ..., Q_N(\bold{s}, \bold{w}_N))$. This definition of the objective function $W$ is general enough to include the objective functions deployed in many existing works. For example, one can use the weighted sum of utilities of all sensors: $
W(Q_1(\bold{s}, \bold{w}_1), ..., Q_N(\bold{s}, \bold{w}_N)) = \sum\limits_{n=1}^N \alpha_n Q_n(\bold{s}, \bold{w}_n)$ where $\{\alpha_n\}_{n\in\{1,..., N\}}$ are weights satisfying $\alpha_n \in [0,1]$ and $\sum\limits_{n=1}^N \alpha_n = 1$. The resource allocation problem for the LPBR can then be formally defined as the derivation of the optimal RAB vector $\bold{s}^*$  that maximizes the objective function of the sensor's utilities:
\begin{equation}
\begin{array}{l}
\bold{s}^* = \arg\max\limits_{\bold{s}} \{W(Q_1(\bold{s}, \bold{w}_1), ..., Q_N(\bold{s}, \bold{w}_N))\}\\
\textrm{subject to}~~ \|\bold{s}\|_0 \leq T
\end{array}\label{eq4}
\end{equation}
where $\|\bold{s}\|_0$ represents the number of non-zero elements in vector $\bold{s}$. Since there are $N^T$ possible RAB vectors $\bold{s}$, exhaustive search for the solution to (\ref{eq4}) is clearly very complex even for modest values for $N$ and $T$. Thus, we try to solve this problem in an alternative way. From (\ref{eq2}) we can see that the slot allocation $\bold{s}$ influences the  $n^{\textrm{th}}$ sensor's utility through term $\sum\limits_{t=1}^T w_{n,t} \textbf{1}(s(t)=n)$. We define this term as the {\it weighted sum rate} and write
\begin{align}
r_n = \sum\limits_{t=1}^T w_{n,t} \textbf{1}(s(t)=n)\label{eq5}
\end{align}
The weighted sum rate can be interpreted as the expected number of packets that can be sent before their corresponding TTS deadlines expire during each RAB. Thus, the $n^{\textrm{th}}$ sensor's utility can be represented by $Q_n(r_n) = \bar{q}_n h_n r_n$. Our proposed approach requires determining the optimal weighted sum rate vector $\bold{r}^*$ first and then finding a resource allocation vector $\bold{s}^*$ that achieves $\bold{r}^*$ within the RAB.

\vspace{-10pt}
\section{Delay-Aware Resource Allocation}
In this section, we provide the slot allocation solution for multi-camera video transmissions with limited information. The proposed Delay-Aware Resource Allocation (DARA) approach comprises two steps. The first step determines the weighted sum rate allocation $\bold{r}$  using the deadline distributions $\bold{w}_n$ and the expected video delivery utility $\bar{q}_n$ of the video bitsteams. In the second step, the slot allocation $\bold{s}$ is determined using the DARA algorithm\footnote{Note that we use ``DARA approach'' to refer to the whole approach and ``DARA algorithm'' to refer to the algorithm in the second step of the DARA approach.} in order to achieve the rate allocation $\bold{r}$. Because the video characteristics (i.e. deadline distributions) are changing over time, the slot allocation is done for each RAB. We note that deadlines and delays do not need to be aligned with the RAB duration since the importance of slots is not affected - sending packets in early slots will satisfy the deadline requirement better than in later slots.

Note that each RAB vector $\bold{s}$  corresponds to a weighted sum rate allocation vector $\bold{r}$ according to (\ref{eq5}). However, the inverse is not true - it can be the case that there is no RAB vector $\bold{s}$ can achieve a given $\bold{r}$. Hence, only certain values of the weighted sum rate can be achieved by the slot allocation, depending on the discounting weights. We write $\mathcal{B}(\{\bold{w}_1,...,\bold{w}_N\}; T)$ as the set of achievable weighted sum rate vectors given weights $\bold{w}_1,...,\bold{w}_N$ and a RAB with $T$ slots.

\vspace{-10pt}
\subsection{Weighted Sum Rate Allocation}
If all sensors are not delay-sensitive, i.e. $\forall n, t: w_{n,t} = 1$, then the slot allocation problem is easy since it reduces to a simple weighted sum rate allocation problem.
\begin{equation}
\begin{array}{l}
\bold{r}^* = \arg\max\limits_{\bold{r}}\{W(Q_1(r_1),...,Q_N(r_N))\}\\
\textrm{subject to}~~~\bold{r}\in \mathcal{B}(\{\bold{w}_1,...,\bold{w}_N;T\})
\end{array}\label{eq7}
\end{equation}
where $\bold{r}\in \mathcal{B}(\{\bold{w}_1,...,\bold{w}_N;T\})$  simply means $\sum\limits_{n=1}^N r_n = T$ and $\forall n: r_n \in \{0,1,...,T\}$.

Since slots have equal value to sensors, it does not matter which specific transmission opportunities are allocated to them within a RAB. Therefore, once the optimal weighted sum rate allocation vector $\bold{r}^*$ is determined, an optimal RAB vector $\bold{s}^*$  (which is obviously not unique) can easily be determined (e.g. using weighted-round-robin schemes or "water-filling" strategies). However, if video transmissions are delay-sensitive, the positions of slots will have different impacts on the expected utility of each sensor's bitstream since the early transmission opportunities can satisfy the deadline requirements more easily. Therefore, delay sensitivity makes the resource allocation problem significantly more difficult.

Since in most scenarios (\ref{eq7}) is a convex optimization, various efficient optimization methods \cite{Boyd} can be used to solve it to obtain the weighted sum rate allocation vector. However, even if we have determined the optimal weighted sum rate allocation vector $\bold{r}^* \in \mathcal{B}(\{\bold{w}_1,...,\bold{w}_N\};T)$, it is unclear which RAB vector $\bold{s}^*$ corresponds to the optimal utility allocation (i.e. with respect to distortion reduction). To address this, in the next subsection we propose an efficient algorithm to compute the corresponding RAB vector $\bold{s}$ given the weighted sum rate allocation vector $\bold{r} \in \mathcal{B}(\{\bold{w}_1,...,\bold{w}_N\};T)$  determined in the first step.

\vspace{-10pt}
\subsection{Delay-Aware Resource Allocation Algorithm}
Given a weighted sum rate allocation vector $\bold{r}^* \in \mathcal{B}(\{\bold{w}_1,...,\bold{w}_N\};T)$, we want to find a RAB vector $\bold{s}$ that achieves $\bold{r}^*$. This is achieved by the DARA algorithm, presented in Algorithm 1.

In slot $t$, the sensor $n^*$, $1\leq n^* \leq N$, with the largest value of the resource allocation metric, $f^\mu_{n^*} w^\nu_{n^*,t}(\sum\limits_{\tau = t + 1}^T w_{n^*,\tau})^{-\gamma}$,  will be given the transmission opportunity of the current slot, where $\mu, \nu, \gamma \in \mathbb{R}$ are algorithm parameters to trade-off the three components of the metric:
\begin{itemize}
  \item $f_n$ stands for the distance to the target weighted sum rate allocation of sensor $n$  from the current allocated weighted sum rate; sensors with larger $f_n$ should be given the higher priority of transmission since their demands are larger;
  \item  $w_{n,t}$ accounts for the benefit of allocating the current timeslot to sensor $n$; sensors with larger $w_{n,t}$  should be given the higher priority of transmission since these sensors can increase their video quality more if they transmit in the current slot;
  \item $\sum\limits_{\tau=t+1}^T w_{n,\tau}$ represents the urgency of allocating the current slot to sensor $n$; sensors with lower $\sum\limits_{\tau=t+1}^T w_{n,\tau}$  should be given the higher priority of transmission since it becomes more difficult to satisfy their demands in the future.
\end{itemize}
The choices of $\mu,\nu,\gamma\in\mathbb{R}$  will depend on specific deployment scenarios. In general,  $f_n$ becomes more important with larger $\mu$, $w_{n,t}$  becomes more important with larger $\nu$  and $\sum\limits_{\tau=t+1}^T w_{n,\tau}$ becomes important with smaller $\gamma$.

\begin{figure}
\centerline{\includegraphics[scale = 0.8]{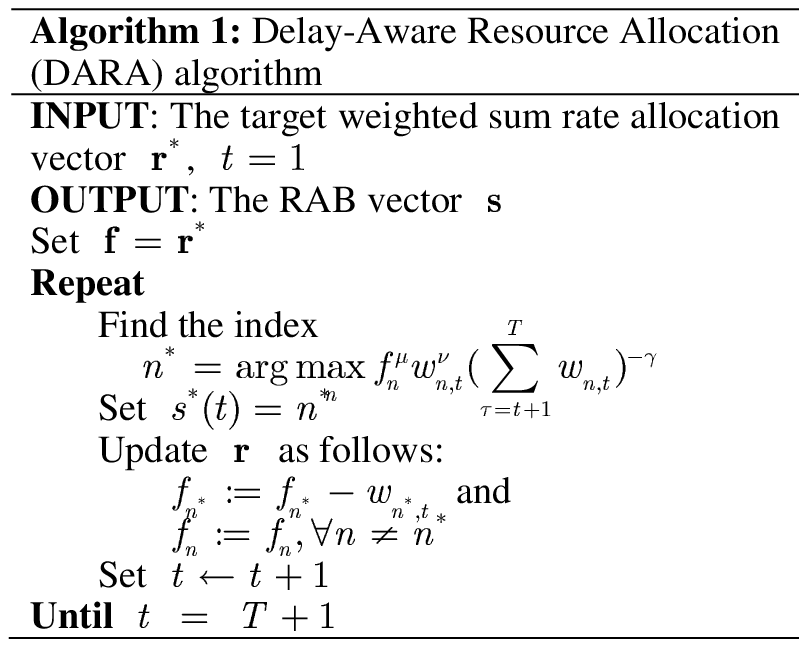}}
\label{algorithm}
\vspace{-20pt}
\end{figure}

The proposed DARA algorithm requires only statistical information of sensors' video transmissions (i.e. the discounting weights) and is simple and easy to implement in practical systems. In the following section, we provide the analytic characterization of the performance of the proposed DARA approach. In particular, we will prove that the proposed DARA approach achieves the optimal weighted sum rate allocation for the exponentially decreasing weights in infinite time horizon and show the performance gap is bounded for finite time horizon.

\vspace{-10pt}
\section{Analytic Characterization of Non-stationary Policies of DARA}
We characterize the performance of the DARA approach in various deployment scenarios. Our analysis focuses on exponentially-decreasing weight distributions\footnote{Our goodness-of-fit tests show that, over a large ensemble of deadline distributions obtained from real video streams (e.g. Figure 2), the exponential distribution achieves the best log-likelihood ratio in comparison to the Half-Normal distribution and the Generalized Pareto distribution.}, that is, for each sensor $n$, there exists a discount factor $\delta_n \in [0,1)$ such that $w_{n,t} = \delta^{t-1}_n$. Since the focus is on the weighted sum rate vector sets, we simplify the notation and write the set of achievable rates as $\mathcal{B}(\{\delta_1,...,\delta_N\};T)$.

\vspace{-10pt}
\subsection{Identical Discount Factors and Infinite Horizon}
Here we assume that sensors have an identical discount factor and RAB duration $T\to \infty$. Under this assumption, we study the set of achievable weighted sum rate vectors and prove the optimality of the DARA approach. Note that in this case, the second and third components of the allocation metric do not affect the resulting resource allocation since their values of all sensors are the same.

The first step to solve the slot allocation problem is to characterize the set of achievable weighted sum rate allocation vectors such that we can determine the optimal $\bold{r}^*$ by solving the optimization problem (\ref{eq7}). When sensors have an identical discount factor, the sum of sensors' weighted sum rate satisfies:
\begin{align}
\sum\limits_{n=1}^N r_n = \sum\limits_{t=1}^\infty \delta^{t-1} = \frac{1}{1-\delta}\label{eq8}
\end{align}
The above condition is necessary but not sufficient for the set of achievable weighted sum rate vectors. For example, suppose sensors are extremely delay-sensitive (i.e. $\delta = 0$), then the set of achievable weighted sum rate vectors includes only the $N$  vectors $\{\forall n \in \{1,...,N\}\}:\bold{r} = (0,...,r_n,...,0)$  where $r_n = 1$. If the condition is also sufficient, then the set of achievable rates is maximized. The following theorem determines when (\ref{eq8}) is also sufficient.

\begin{theorem}
Suppose the weight distributions are exponentially-decreasing, sensors have an identical discount factor $\forall n: \delta_n = \delta$  and the RAB duration $T \to \infty$. We can achieve the following set of weighted sum rate vectors
$\mathcal{B}(\{\delta,...,\delta\};\infty) = \{\bold{r}:\sum\limits_{n=1}^N r_n = \frac{1}{1-\delta}\}$
if and only if the discount factor $\delta \geq 1 - 1/N$.
\end{theorem}
\begin{proof}
Since the weights are $w_{n,t} = \delta^{t-1}$, based on (\ref{eq5}), we write the weighted sum rate  $r_n$ as $r_n = \sum\limits_{\tau = 1}^\infty \delta^{\tau-1}\textbf{1}(s(\tau) = n)$.
We also write $r_n = r_n(1)$, indicating that it is the weighted sum rate {\it calculated at timeslot 1}. In general we can define   as the weighted sum rate calculated at timeslot  . Similar to the Bellman equation in dynamic programming, we can decompose   into the current rate at time slot 1 and the continuation weighted sum rate from timeslot 2:
\begin{equation}
\begin{array}{ll}
r_n(1) &= \sum\limits_{\tau=1}^\infty \delta^{\tau-1} \textbf{1}(s(\tau) = n) \\
&= \textbf{1}(s(1) = n) + \delta\sum\limits_{\tau=2}^\infty \delta^{\tau-1} \textbf{1}(s(\tau)=n)\\
& = \textbf{1}(s(1) = n) + \delta r_n(2)
\end{array}
\end{equation}
Similarly, $r_n(2)$ can be further decomposed into the current rate at timeslot 2 and $r_n(3)$. We call the vector of weighted sum rates $\bold{r}$  in the set $\mathcal{B}(\{\delta,...,\delta\};\infty)$  a feasible vector of weighted sum rates $\bold{r}$. The main idea of the proof is to show that when $\delta \geq 1 - 1/N$, for any $t\geq 1$, any feasible vectors of weighted sum rates $\bold{r}(t) \in \mathcal{B}(\{\delta,...,\delta\};\infty)$ can be decomposed by a vector of current rates $[\textbf{1}(s(t)=1),...,\textbf{1}(s(t)=1)]^T$ and a feasible vector of continuation weighted sum rates $\bold{r}(t+1) \in \mathcal{B}(\{\delta,...,\delta\};\infty)$. As long as the above holds, we can pick any feasible $\bold{r}(1)$, and decompose it to determine the current vector of weighted sum rates at timeslot 1 (i.e. to determine which sensor transmit at timeslot 1). Since all $\bold{r}(1)$  is feasible, we can decompose it to determine the current vector of weighted sum rates at timeslot $t$ (i.e. to determine which sensor transmits at timeslot $t$). In this way, we can obtain the slot allocation policy that yields the vector of weighted sum rates $\bold{r} = \bold{r}(1)$.

For easier exposition of the proof, we normalize $r_n$  by $1-\delta$, yielding the normalized weighted sum rate $v_n$, which represents the weighted average rate of sensor $n$. Consequently, a vector of normalized weighted sum rates $\bold{v}$ is feasible if it is in $\mathcal{B}_v(\{\delta,...,\delta\};\infty) = \{\bold{v}:\sum\limits_{n=1}^N v_n = 1\}$. The decomposition of $\bold{v}(t)$ can be written as
\begin{equation}
\begin{array}{rl}
v_n(t) =& (1-\delta)\sum\limits_{\tau = t}^\infty \delta^{\tau-t} \textbf{1}(s(\tau) = n)\\
=&(1-\delta)\textbf{1}(s(t) = n) \\
&+ (1-\delta) \sum\limits_{\tau = t + 1}^\infty \delta\cdot\delta^{\tau - (t+1)} \textbf{1}(s(\tau) = n)\\
=&(1-\delta)\textbf{1}(s(t) = n) + \delta v_n(t+1)
\end{array}
\end{equation}

Hence, our goal is to show that when $\delta \geq 1 - 1/N$, for any $t\geq 1$ any feasible vector of weighted sum rates $\bold{v}(t) \in \mathcal{B}_v(\{\delta,...,\delta\};\infty)$  can be decomposed into a vector of current rates $[\textbf{1}(s(t)=1),...,\textbf{1}(s(t)=N)]^T$  and a feasible vector of continuation normalized weighted sum rates $\bold{v}(t+1) \in \mathcal{B}_v(\{\delta,...,\delta\};\infty)$.

Suppose time slot $t$ is allocated to sensor $n$. Then $\bold{v}(t)$ can be decomposed as follows,
\begin{equation}
\begin{array}{c}
v_n(t) = (1-\delta) + \delta v_n(t+1)\\
\forall m \neq n: v_m(t) = \delta v_n(t+1)
\end{array}
\end{equation}
The continuation normalized weighted sum rates at time $t+1$ can therefore be derived as $v_n(t+1) = \frac{v_n(t) -1 +\delta}{\delta}$, $v_m(t+1) = \frac{v_m(t)}{\delta}$.
It can be easily verified that $\sum\limits_{n=1}^N v_n(t+1) = 1$ is automatically satisfied. However, for $\bold{v}(t+1)$ to be feasible, we also need $\forall n: v_n(t+1) \geq 0$. This requires $\delta \geq 1 - v_n(t)$. Therefore, the minimum discount factor is $\underline{\delta} = \max\limits_{\bold{v} \in \mathcal{B}_v} \min\limits_{n} \{1 - v_n\}$.
This is achieved when $\bold{v} = \{1/N,...,1/N\}$. Therefore $\delta = 1 - 1/N$.	
\end{proof}

Theorem 1 states when the discount factor is larger enough, i.e. video transmission is not very delay-sensitive, the set of achievable weighted sum rate vectors can be maximized. In the other extreme case, when $\delta = 0$, only $N$ finite weighted sum rate allocation vectors can be achieved as we mentioned earlier.

The result of Theorem 1 is important for the resource allocation design problem since we need to understand what can possibly be achieved by the DARA approach (since the first step of the DARA approach needs to determine the target weighted sum rates $\bold{r}^*$). When $\delta \geq 1- 1/N$, we can simply obtain the optimal allocation $\bold{r}^*$ by solving the following optimization problem
\begin{equation}
\begin{array}{c}
\bold{r}^* = \arg\max\limits_{\bold{r}} W(Q_1(\bold{r}),...,Q_N(\bold{r}))\\
\textrm{subject to}~~~\sum\limits_{n=1}^N r_n = \frac{1}{1-\delta}
\end{array}
\end{equation}
This optimization problem is easy to solve when $W$ is a convex function in $(r_1,...,r_N)$. For example, if the objective is to maximize the minimum of the weighted sensors' utilities, i.e. $W=\min\limits_n \alpha_n \bar{q}_n h_n r_n$, then the solution can be obtained analytically as $r^*_n = \left[(1-\delta)\sum\limits_{i=1}^N \frac{\alpha_n \bar{q}_n h_n}{\alpha_i \bar{q}_i h_i}\right]^{-1}$. After the optimal allocation $\bold{r}^*$ is determined, we then can run the DARA algorithm to find the optimal resource allocation vector $\bold{s}$. When sensors have an identical discount factor, in each slot, finding the sensor $n^*$, $1\leq n^* \leq N$, that maximizes $f_n^\mu w_{n,t}^\nu(\sum\limits_{\tau=t+1}^T w_{n,\tau})^{-\gamma}$ is equivalent to finding the sensor with the maximum $f_n$. The following theorem proves that the DARA algorithm is able to achieve any weighted sum rate vector in the set $\mathcal{B}$ and hence, the optimal $\bold{s}^*$ can be determined.

\begin{theorem}
Suppose the weight distributions are exponentially-decreasing, sensors have an identical discount factor $\forall n:\delta_n = \delta$  and the RAB duration $T \to \infty$. For any target weighted sum rate vector $\bold{r}\in\mathcal{B}(\{\delta,...,\delta\};\infty)$ and any  $\delta\geq 1 - 1/N$, the resource allocation $\bold{s}$ generated by running the DARA algorithm achieves $\bold{r}$.
\end{theorem}
\begin{proof}
We have proved in Theorem 1 that, when  $\delta \geq 1- 1/N$, at any timeslot $t$, any feasible $\bold{r}\in\mathcal{B}(\{\delta,...,\delta\};\infty)$ can be decomposed into a vector of slot allocation $[\textbf{1}(s(t) = 1),...,\textbf{1}(s(t)=N)]^T$ and the continuation weighted sum rate $\bold{r}(t+1) \in \mathcal{B}(\{\delta,...,\delta\};\infty)$ as follows: (when $s(t) = n$)
\begin{equation}
r_n(t) = 1 + \delta r_n(t+1) ; \forall m\neq n: r_m(t) = \delta r_m(t+1)
\end{equation}
This decomposition determines which sensor to transmit at timeslot $t$. In the DARA algorithm, we start with setting the target weighted sum rate vector $\bold{r}=\bold{r}(1)$ at slot 1, which is then decomposed to determine which sensor to transmit in timeslot 1. Then we decompose the continuation weighted sum rate vector and determine which sensor should transmit in timeslot 2. By performing the decomposition in every timeslot of the RAB, we can determine which sensor should transmit in every timeslot. Based on the proof of Theorem 1, we can do this decomposition forever because every continuation weighted rate vector is feasible. Moreover, the weighted sum rate achieved as $T\to \infty$ will be $\bold{r}$.
\end{proof}

\vspace{-10pt}
\subsection{Finite Time Horizon}
In practice, the channel coherence time $T$ is not infinite. Hence, the RAB duration $T$ is finite. In the following, we characterize the performance gap of the proposed algorithm when $T$ is finite instead of infinite. Let $v^T_n = \frac{r^T_n}{\sum\limits_{t=1}^T \delta^{t-1}}$, $v^\infty_n = \frac{r^\infty_n}{\sum\limits_{t=1}^\infty \delta^{t-1}}$
be the normalized weighted sum rate allocation when the time horizon is finite and infinite, respectively. The following proposition derives an upper bound for the distance of the achieved weighted sum rate allocation $v^T_n$ (by running the algorithm for only $T$ slots) from the optimal normalized weighted sum rate $v^\infty_n$.

\begin{proposition}
$|v^T_n - v^\infty_n|\leq \delta^T$.
\end{proposition}
\begin{proof}
We can write $v^T_n$ and $v^\infty_n$ as:
\begin{align}
v^T_n = \frac{\sum\limits_{t=1}^T \delta^{t-1} \textbf{1}(s(t)=n)}{\sum\limits_{t=1}^T \delta^{t-1}}~\textrm{and}~v^\infty_n = \frac{\sum\limits_{t=1}^\infty \delta^{t-1}\textbf{1}(s(t)=n)}{\sum\limits_{t=1}^\infty \delta^{t-1}}
\end{align}
Finally, since,
\begin{equation}
\begin{array}{cc}
v^T_n  \geq  (1-\delta)\sum\limits_{t=1}^T \delta^{t-1} \textbf{1}(s(t)=n)\\
v^\infty_n \leq (1-\delta)\sum\limits_{t=1}^T \delta^{t-1} \textbf{1}(s(t)=n) + \delta^T
\end{array}
\end{equation}
we thus have $v^\infty_n - v^T_n \leq \delta^T$.
Secondly, we have
\begin{equation}
\begin{array}{c}
v^\infty_n - v^T_n \geq \frac{\sum\limits_{t=1}^T \delta^{t-1}\textbf{1}(s(t)=n)}{\sum\limits_{t=1}^\infty \delta^{t-1}} - \frac{\sum\limits_{t=1}^T \delta^{t-1} \textbf{1}(s(t)=n)}{\sum\limits_{t=1}^T \delta^{t-1}} \\
= -(1-\delta) \frac{\delta^T}{1-\delta^T}\sum\limits_{t=1}^T \delta^{t-1} \textbf{1}(s(t) = n)
\end{array}
\end{equation}
Since $\sum\limits_{t=1}^T \delta^{t-1} \textbf{1}(s(t) = n) \leq \sum\limits_{t=1}^T \delta^{t-1} = \frac{1-\delta^T}{1-\delta}$, we have
\begin{equation}
\begin{array}{c}
v^\infty_n - v^T_n \geq -(1-\delta) \frac{\delta^T}{1-\delta^T}\sum\limits_{t=1}^T \delta^{t-1} \textbf{1}(s(t) = n) \\
\geq -(1-\delta) \frac{\delta^T}{1-\delta^T}\frac{1-\delta^T}{1-\delta}= -\delta^T
\end{array}
\end{equation}
In summary, we have $-\delta^T \leq v^\infty_n - v^T_n \leq \delta^T$, which means $|v^T_n - v^\infty_n|\leq \delta^T$.
\end{proof}
Proposition 1 proves that the performance gap by limiting the slot (i.e. time) horizon to be finite. Moreover, the gap can be made arbitrarily small by choosing $T$  large enough. Therefore, the DARA algorithm asymptotically achieves the optimal weighted sum rate vector determined in the first step of the DARA approach as the number of slots goes to infinity.

\vspace{-10pt}
\subsection{Resource Allocation for the General Case}
In this subsection, we study the resource allocation problem when sensors have heterogeneous discounting weights and when the time horizon is finite. It is analytically difficult to determine the exact set of achievable weighted sum rate vectors. Hence, we will instead use approximate sets to solve the optimization problem (\ref{eq7}).

We notice that every feasible allocation vector $\bold{r}$ is bounded by $\sum\limits_{t=1}^T\min\limits_n w_{n,t} \leq \sum\limits_{n=1}^N r_n \leq \sum\limits_{t=1}^T \max\limits_n w_{n,t}$. Hence, an approximate set of achievable weighted sum rate vectors can be $\tilde{\mathcal{B}}(\{\bold{w}_1,...,\bold{w}_N\};T) = \{\bold{r}:\sum\limits_{n=1}^N r_n = R\}$ with $\sum\limits_{t=1}^T\min\limits_n w_{n,t} \leq R\leq \sum\limits_{t=1}^T \max\limits_{n} w_{n,t}$. For fixed $R$, when the objective function $W$  is convex in $\{r_1,...,r_N\}$, the optimization problem (\ref{eq7}) is convex and it is thus solvable in polynomial time. For example, if the objective function is the minimum of the weighted sensors' utilities, namely  , then the solution can be obtained analytically as
$r^*_n = \frac{R}{\sum\limits_{i=1}^N\frac{\alpha_n \bar{q}_n h_n}{\alpha_i \bar{q}_i h_i}}$.
Once the approximate optimal weighted sum rate vector $\bold{r}^*$ is obtained, we can construct the slot allocation vector $\bold{s}$ using the DARA algorithm. The performance of using the approximation set and the DARA algorithm will be evaluated in the next section.

\vspace{-10pt}
\section{Experiments}
In this section, we evaluate the performance of the proposed resource allocation solution via: (i) numerical studies; (ii) experiments using the IETF 6tisch simulation of the IEEE 802.15.4e TSCH MAC \cite{Watteyne}\cite{Wang} instantiation of the IEEE 802.15.4e TSCH MAC and standard surveillance videos encoded with H.264/AVC under the assumption of 4, 6 and 10 visual sensors sharing the IEEE 802.15e TSCH timeslots. These were found to be representative cases under the bandwidth constraints of the physical layer of IoT-oriented standards \cite{Watteyne}\cite{Wang}\cite{IEEE 802.15.4}.

\vspace{-10pt}
\subsection{Benchmarks}
Three benchmark policies are used in our experiments:

\begin{itemize}
  \item (i) round-robin slot assignment (``Round-robin''), which is the most typical option in WPAN MAC designs \cite{Watteyne}\cite{Wang}\cite{Bachir}\cite{Pister};
  \item (ii) rate-proportional round-robin slot assignment (``R-Round-robin'') \cite{Dutta}, where the average bit-budget of each sensor within each RAB is used to allocate a proportional number of slots in a round-robin fashion (higher budget corresponds to more slots for a sensor);
  \item  (iii) rate/delay-proportional round-robin slot assignment (``RD-Round-robin'') \cite{Pradas}, where, beyond the bit-budget, the delay deadline of each sensor is also used to weight the slot assignment according to a heuristic rule, i.e., beyond rate, the delay deadline is taken into account in the proportional slot allocation.
\end{itemize}

Note that the optimal slot allocation ("Optimal"), i.e. the solution to (\ref{eq4}) that is computationally hard: since we can choose any one of the $N$ sensors in each one of the $T$ timeslots in a RAB, the total number of possible allocations is $N^T$. For example, even with 2 sensors and a RAB with 100 timeslots, we need to search amongst more than $10^{30}$ possible slot allocations.

\begin{table}
\caption{Comparison of different policies in increasing order of complexity.}\label{table4}
\vspace{-10pt}
\centerline{\includegraphics[scale = 0.8]{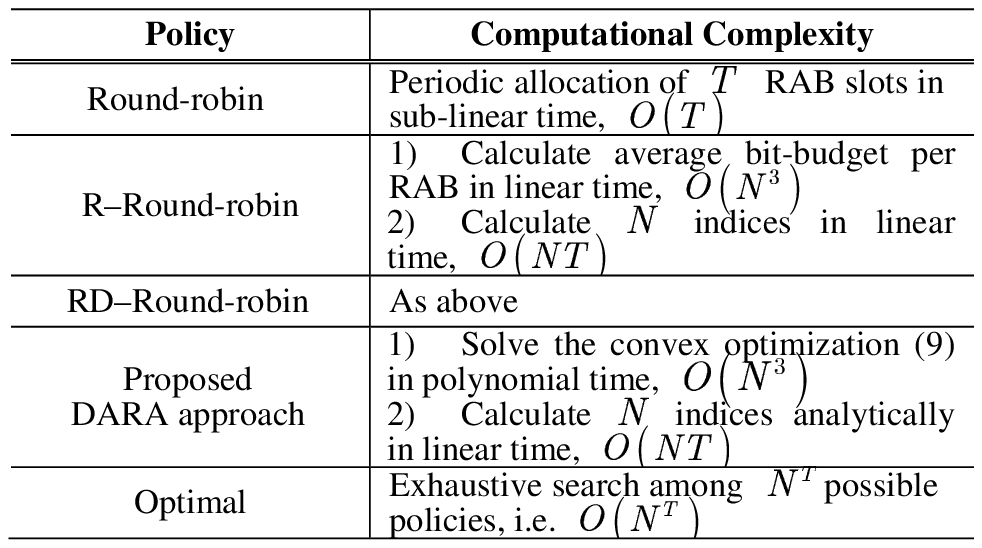}}
\vspace{-15pt}
\end{table}

Table IV lists the computational requirements of each policy, including the proposed DARA approach. The first steps of the R-Round-robin, RD-Round-Robin and the proposed DARA approach involve solving convex optimization problem which can be solved efficiently by many existing algorithms \cite{Boyd}. The second steps are more critical which determines the slot assignment among the sensors. Since the number of visual sensors is limited (i.e. in the vast majority of cases $N < 20$), the complexity of determining the weighted sum rate allocation (i.e. the first step of our DARA approach) is low. Moreover, since our DARA algorithm ensures that the complexity of determining the slot allocation (i.e. the second step of the DARA approach) is linear in $T$, the overall computation complexity is not a big concern. Note that performing the exhaustive search to determine the optimal slot allocation is extremely complex, i.e. $O(N^T)$, which prevents it from being implementable in practical systems.

\vspace{-10pt}
\subsection{Numerical Study}
In this subsection, we present numerical results with utility functions for each sensor that correspond to Theorems 1 and 2 and also use normalized scaling for the expected utilities. While this study does not directly map to real video sequences and an IEEE 802.15.4e network, it establishes the validity of the proposed DARA approach as a general non-stationary multi-user resource allocation method in a ``bias-free'' manner, i.e. regardless of the specific settings of the WPAN and video codec used. Thus, beyond the specific context of this paper, via such numerical studies one can extrapolate the usefulness of the DARA approach for other cases, e.g. in streaming of audio or other error-tolerant multi-sensor data volumes under delay constraints.

We assume that each sensor produces video traffic whose MAC frames (packets) can be characterized by a normal distribution \footnote{Similar numerical experiments have been derived for other distributions and other settings - we only illustrate this case as a representative one.} with mean $200$  and standard deviation $20$, i.e. $\forall n: h_n \sim \mathcal{N}(200,20)$. The expected video delivery utility per sensor is normalized to unity, i.e. $\forall n: \bar{q}_n = 1$. Under these conditions, we aim to maximize the minimum utility amongst all sensors, namely
\begin{equation}
\begin{array}{c}
W = \min\limits_{n} \alpha_n \bar{q}_n h_n r_n\\
\textrm{subject to}~~~\bold{r}\in\mathcal{B}(\{\bold{w}_1,...,\bold{w}_N\};T)
\end{array}\label{eq31}
\end{equation}
where $\forall n: \alpha_n = 1/N$. Unless otherwise stated, the DARA parameters are set to be  $\mu = \nu = \gamma = 1$.

\subsubsection{Identical Discount Factors}
We first investigate the scenarios where sensors have identical discount factors. In Figure 3, the utility achieved by our proposed DARA approach and the two benchmarks (R-Round-robin and RD-Round-robin coincide for identical discount factors) are shown for various numbers of sensors $N \in \{2,...,10\}$ and two discount factors $\delta_{\textrm{low}} = 0.99$  and $\delta_{\textrm{high}} = 0.995$. The resource allocation block size is set to $T=500$ slots and, for ease of illustration, the minimum utility of (\ref{eq31}) has been normalized to   for each algorithm.  As   increases, each sensor is able to obtain less timeslots and its utility is thus decreasing. Moreover, the value of later transmission opportunities is higher for larger values of the discount factor   and thus the utility is higher for such cases. The proposed DARA approach significantly outperforms the other benchmarks.

Since the DARA algorithm is adaptively changing the slot assignment in order to achieve the target objective function of (\ref{eq31}) under Algorithm 1, i.e. maximizing the minimum utility amongst all sensors, it is important to study whether the algorithm can actually achieve the target objective. Given this max-min objective, the target utilities for all sensors are determined to be the same for all sensors, i.e. 52.9, in the first step of the DARA approach. Table V illustrates the utilities of individual sensors achieved by the proposed DARA approach, Round-Robin and R-Round-Robin. The network size is set to be $N=6$, the discount factor is  $\delta = 0.99$ and the RAB is set to $T=500$. As we can see, the utilities achieved by running the DARA algorithm are very close to the target utilities. However, both Round-Robin and R-Round-Robin do not achieve the target utilities that maximize the minimum utility amongst all sensors.

\begin{figure}
\centering{\includegraphics[scale = 0.5]{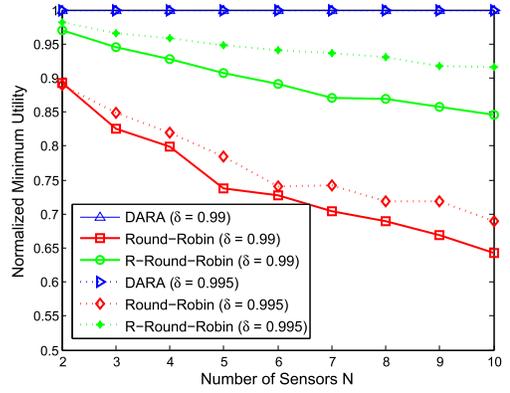}}
\caption{Performance comparison of various approaches under identical discount factors.}\label{homo}
\vspace{-20pt}
\end{figure}

\subsubsection{General Discounting Factors}
When the discounting factors are not exponentially decreasing or identical for all sensors, the proposed DARA approach may not achieve the optimal performance. Figure 4 illustrates the max-min utility of sensors (normalized to $[0,1]$) achieved by all approaches for two sets of discounting weights and   sensors in the network. In the first set of simulations (solid curves), the sensors' discounting factors are selected from $[0.990, 0.992]$ with equal intervals. In the second set of simulations (dashed curves), the sensors' discounting factors are selected from $[0.995,0.997]$ with equal intervals. This set of simulations serve as a counterpart for the case where sensors have identical discount factors illustrated in Figure 3. When solving (\ref{eq7}), we use $R = \sum\limits_{t=1}^T \min\limits_n w_{n,t}$. DARA again significantly outperforms Round-robin, R-Round-robin, and R-D-Round-robin. Importantly, the max-min utilities achieved by DARA are similar to those in the case of identical discount factors. The latter fact indicates that the deviation from optimality does not bring substantial penalties for DARA.

\begin{figure}
\centering{\includegraphics[scale = 0.5]{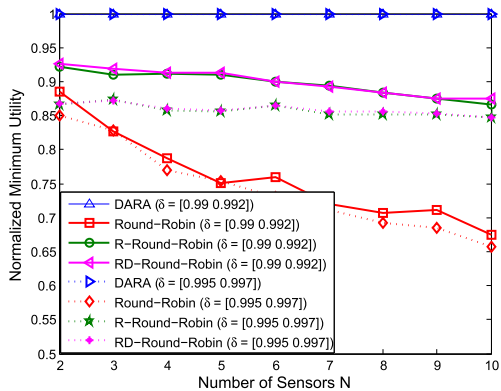}}
\caption{Performance comparison of various approaches under different discount factors.}\label{heter}
\vspace{-10pt}
\end{figure}

\begin{table}
\caption{Achieved individual utilities by DARA, Round-Robin and R-Round-Robin.}\label{Table5}
\vspace{-10pt}
\centering{\includegraphics[scale = 0.7]{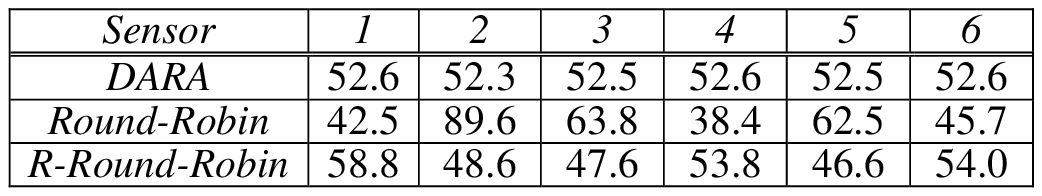}}
\vspace{-15pt}
\end{table}

\vspace{-10pt}
\subsection{Application in Video Transmission under IEEE 802.15.4e TSCH MAC}
To validate the proposed DARA approach under realistic conditions, we used the IETF 6tisch simulator (which is associated with the on-going Berkeley openWSN implementation for the IEEE 802.15.4e \cite{Watteyne}\cite{Wang}) and can be deployed on TelosB motes with the CC2420 transceiver. Several camera sensors can be coupled with this setup (e.g. MEMSIC's Lotus or IMB400, or EURESYS's Picolo V16-H.264) in order to capture and encode video content with the H.264/AVC codec and encapsulate it in MP4 or MKV format. However, in order to provide experiments under controlled conditions and with well-established and publicly-available content, we instead encoded surveillance videos from the http://pets2007.net/ website (PETS 2007 benchmark data) after converting them to: (i) 4 views of CIF resolution at 10 Hz; (ii) 10 views of QCIF resolution at 4 Hz. These can then be packetized and streamed by each openWSN-based TelosB mote under a pre-established transmission delay deadline for each video.

\subsubsection{Content Generation}
Out of a wide range of experiments performed, we present a representative case using dataset S1 from the PETS benchmark. This dataset comprises 2 min 20 sec of 720¡Á576 RGB 25 Hz video stemming from four different cameras of a public area. This led to four CIF-10Hz video feeds by filtering and downsampling and frame downconversion, which correspond to $N=4$ sensors. For $N=6$ and $N=10$ sensors, the four feeds where further downscaled to QCIF-4Hz, and two and six additional video feeds where created by cropping a QCIF-4Hz section of the original videos at spatial positions $(100,100)$ and $(250,250)$, respectively. By using H.264/AVC encoding in low-delay/quantization-stepsize mode, average video bitrates between 25-50 kbps where achieved per visual sensor for the CIF-10Hz case, while 4-13kbps where achieved for the QCIF-4Hz case.

\subsubsection{System Description}
These $N = \{4,6,10\}$ views were streamed via the 6tisch simulation of the TSCH MAC of IEEE 802.15.4e. For our experiments, each bitstream was encapsulated with the MKV container, as: (i) we have already developed low-delay open-source streaming mechanisms in our prior work \cite{Tsakos} that derive the hinting description of the content in real time and are tolerant to a very wide range of packet loss rates; (ii) our prior work \cite{Tsakos} can handle the streaming service beyond the WPAN within the IoT framework of Figure 1, thus allowing for the provisioning of an end-to-end IoT-based multi-camera video streaming solution. The utilized hinting description contains the transmission deadlines for each video frame (and consequently for each MAC timeslot) as well as its size and frame type ($\{I, P, B\}$). Thus, each sensor only needs to generate its weight distributions based on the hinting description of its recent content and communicate them, along with the other three parameters $\alpha_n$, $\bar{q}_n$, $h_n$,  to the LPBR once every 12 slotframes. The LPBR solves the optimization problem based on the parameters  $\alpha_n$, $\bar{q}_n$, $h_n$  to obtain the optimal weighted sum rate vector, and then derives the slot allocation via Algorithm 1, which is communicated to all sensors in downlink mode in order to be used for the subsequent 12 slotframes.

\subsubsection{Experimental Settings and Video Quality Characterization}
The utilized settings are representative to a video surveillance or monitoring application where the WPAN-based monitoring could allow for the content to be streamed over an IoT-enabled application to remote users having a variety of devices connected via a variety of networks. The parameters for our experiments were: $N = \{4, 6, 10\}$, $T_{\textrm{SF}}=1$ s, $T_{\textrm{slot}} = 7.7$ ms,  $h_n = 1$, $b_{\textrm{packet}} = 110$ bytes, $d_n\in [50, 600]$ ms (set randomly for each sensor), VideoLAN x264 encoder with configuration [--preset placebo --tune psnr --muxer mkv --keyint {40,16} --crf {44,35}], i.e. one I frame every 4 seconds and two crf values (corresponding to quantization stepsize). These parameters led to variable-bitrate encoding with average rates of 25-60 kbps for CIF-10Hz and 4-13 kbps for QCIF-4Hz, respectively. For the utilized content, application context, and encoding settings, we have verified via visual inspection that:
\begin{itemize}
  \item (i) mean PSNR between 20~25 dB corresponds to useful-quality video, i.e. blurriness and frame drops occur in several sections but main features and motion of objects or people is discernible;
  \item (ii) mean PSNR between 25~30 dB corresponds to good-quality video (low blurriness and almost no loss of motion or objects' characteristics);
  \item (iii) mean PSNR above 30 dB corresponds to high-quality video (video looks almost like the original albeit with very minor artifacts).
\end{itemize}

\subsubsection{Experiments}
Clearly, the results will vary depending on the delay deadline imposed on each sensor as well as the number of sensors in the WPAN. However, one other aspect that is important in the achievable performance is the difference between the weight distributions of the sensors. For instance, the DARA approach will have the maximum benefit from diversity in the weight distributions as it will assign the RAB slots according to the delay tolerance and the overall bit budget of each sensor.  The degree of diversity in these weight distributions depends on: (i) the transmission delay deadlines imposed; (ii) whether some (or all) sensors have their intra frames in synchronized transmissions. The prominence of intra frames is due to their large size in comparison to the remaining frames, which causes a large peak in the weight distribution at the particular RAB where their appear, as shown in Figure 2. We thus present three separate sets of results. Specifically, in Figure 5 and Figure 6 the average PSNR per sensor (luminance channel) is presented when imposing minimum diversity (i.e. worst-case scenario where all I-frames of all sensors are temporally aligned) and maximum diversity in the sensor transmissions (i.e. I frames of sensors are as temporally misaligned as possible), respectively. As expected, the results demonstrate that, under maximum I-frame alignment (worst case scenario), all algorithms achieve similar performance: they will only be able to accommodate only 1~5 sensors with mean PSNR above 20 dB), which means that the video of the majority of sensors is of unusable quality.

\begin{figure}
\centerline{\includegraphics[height = 1.6in, width = 3.4in]{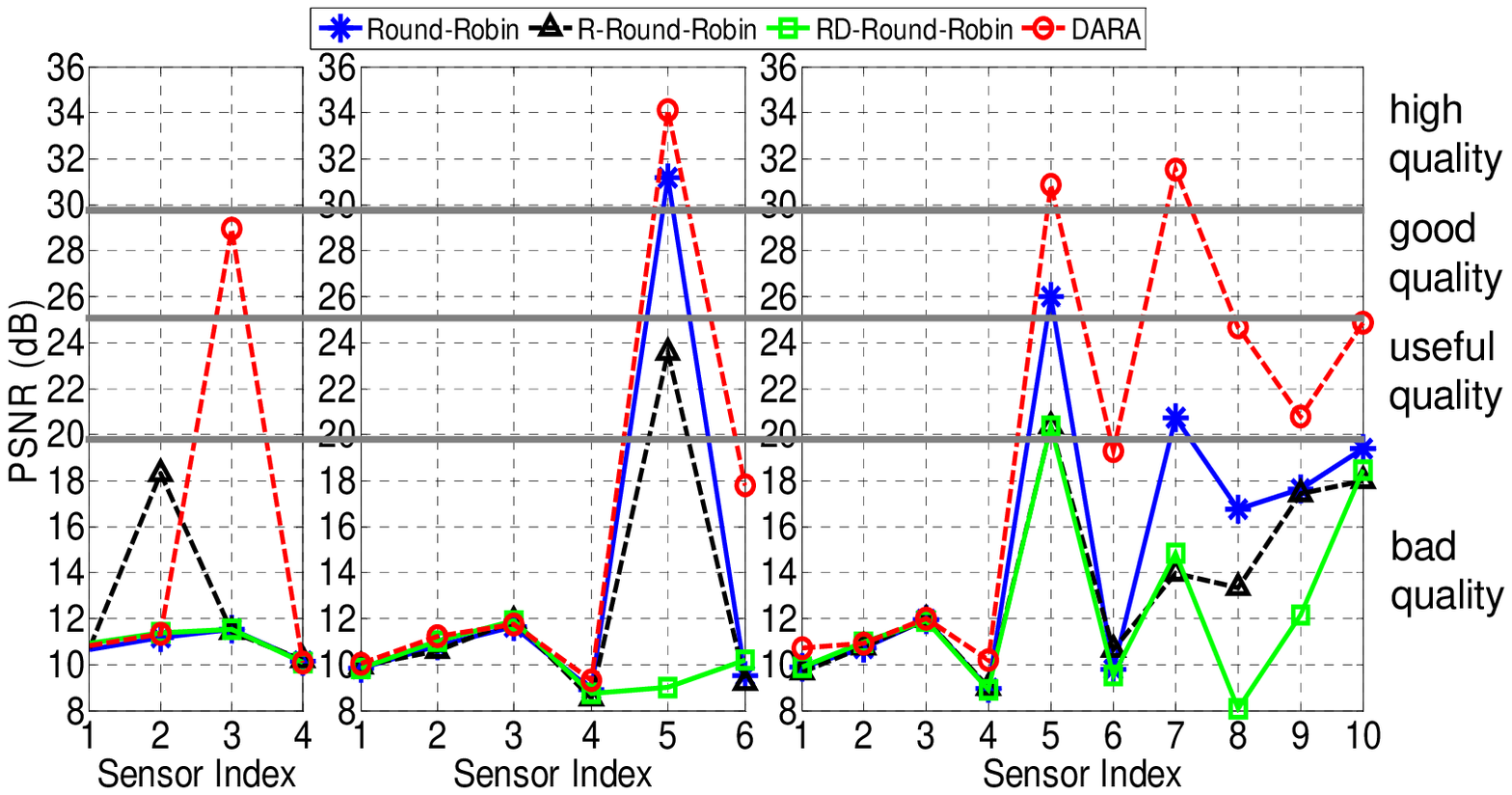}}
\vspace{-10pt}
\caption{Worst-case results: Mean PSNR (luminance channel, dB) per sensor for different methods (averaged over time and over several experiment repetitions) under maximum alignment. }\label{worstcase}
\vspace{-10pt}
\end{figure}

\begin{figure}
\centerline{\includegraphics[height = 1.6in, width = 3.4in]{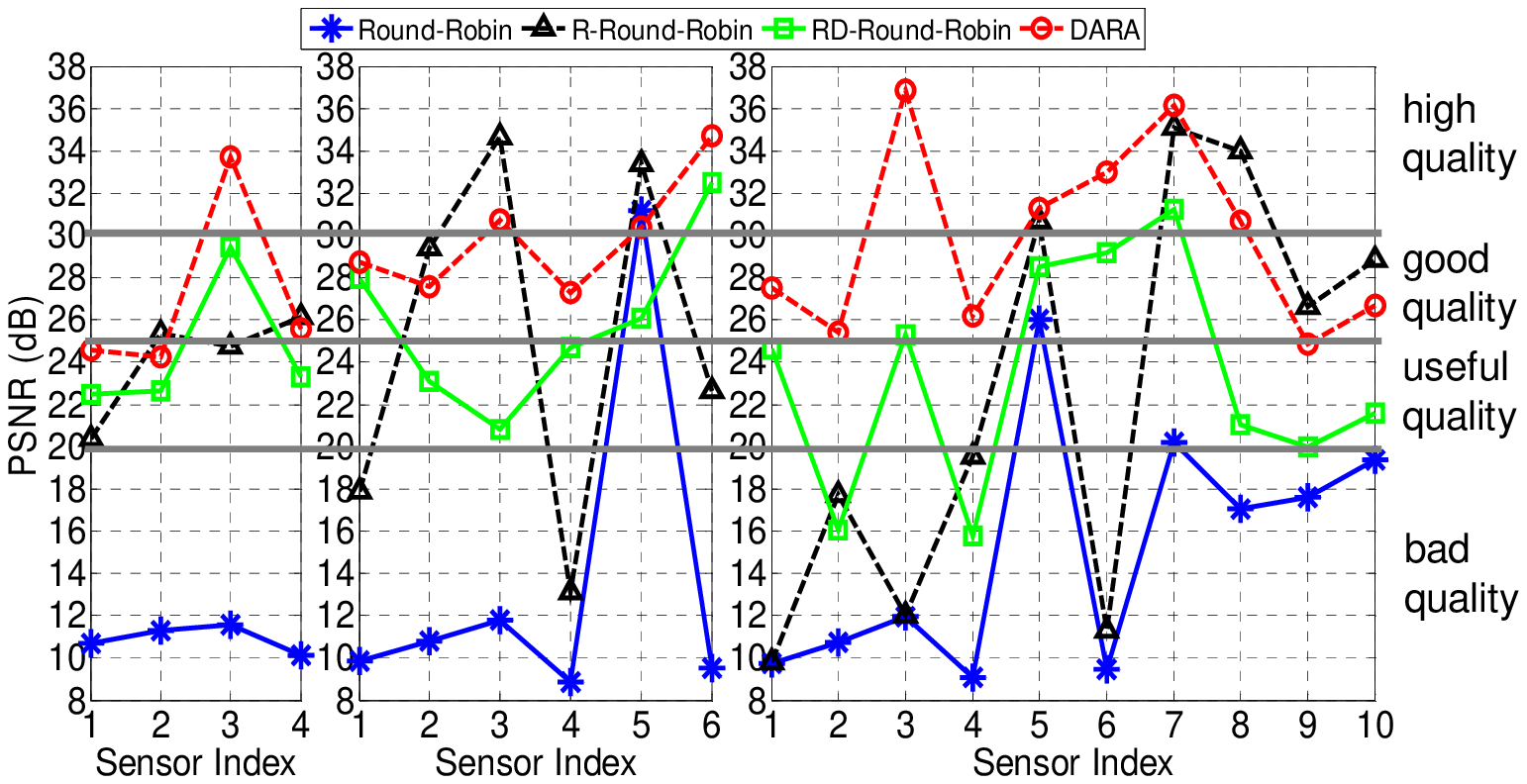}}
\vspace{-10pt}
\caption{Best-case results: Mean PSNR (luminance channel, dB) per sensor for different methods (averaged over time and over several experiment repetitions) under maximum misalignment. }\label{bestcase}
\vspace{-15pt}
\end{figure}

\begin{figure}
\centerline{\includegraphics[height = 1.6in, width = 3.4in]{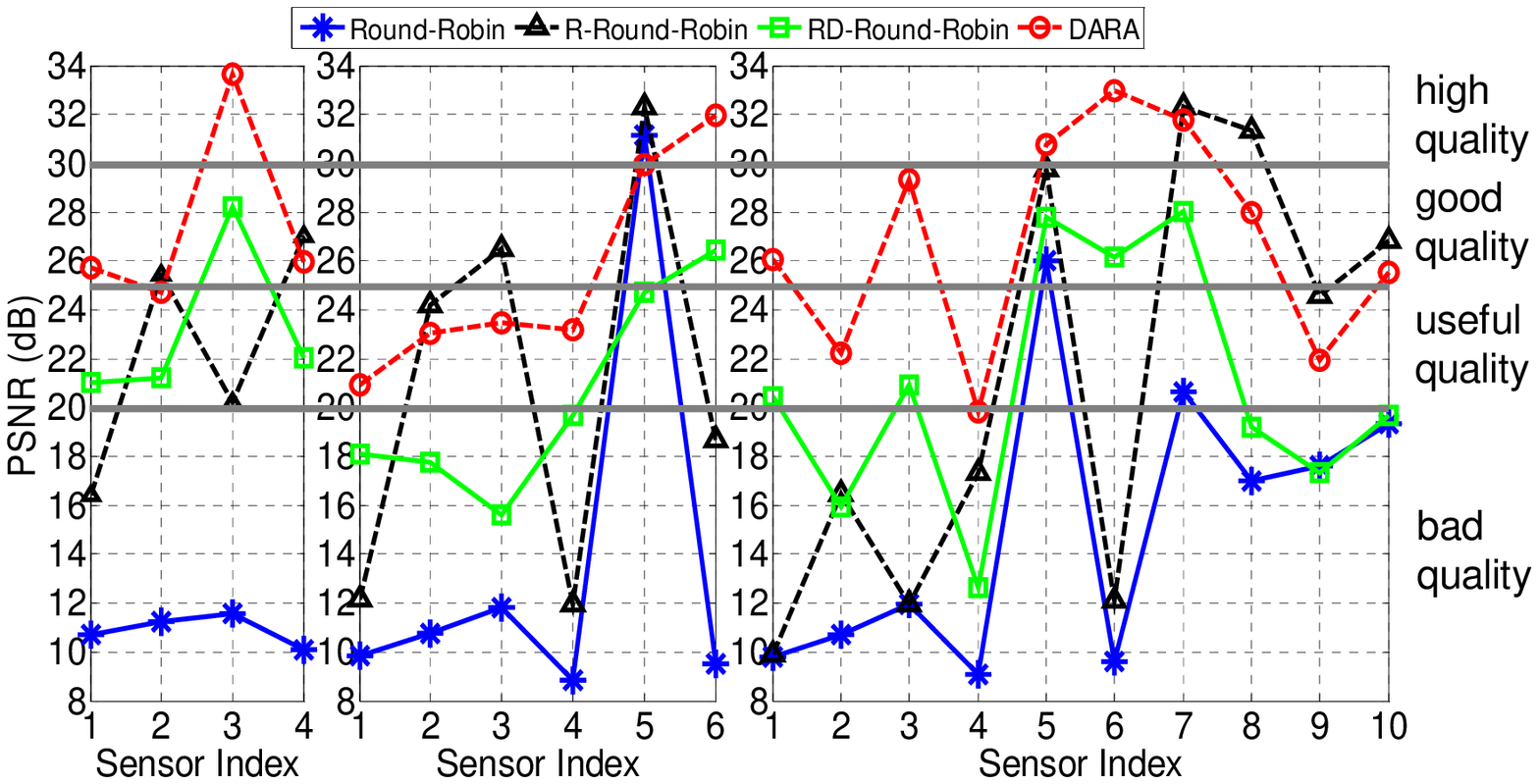}}
\vspace{-10pt}
\caption{Average-case results: Mean PSNR (luminance channel, dB) per sensor for different methods averaged over: time, experiment repetitions and all possible misalignments. }\label{averagecase}
\vspace{-15pt}
\end{figure}

Under maximum misalignment (i.e. best-case scenario of Figure 6), Round-robin continues to fail as its PSNR is almost always below 20 dB per sensor and R-Round-robin and RD-Round-robin achieve moderate success in transmitting useful-quality video (i.e. PSNR above 20 dB). On the other hand, the DARA algorithm not only succeeds in transmitting useful-quality video from all sensors, but also achieves mean PSNR above 25 dB (i.e. good quality) for all sensors. This indicates that, even under the best-case scenario, the DARA algorithm is the only method that can ensure the deadline-abiding transmission of all video bitstreams with useful quality.

In practice, the video bitstreams of different sensors have random alignment conditions. Thus, Figure 7 presents the mean PSNR averaged under random alignment conditions in the IEEE 802.15.4e TSCH MAC slotframe assignment. Similarly, as before, while all other algorithms have one or more sensors with average PSNR below 20 dB (i.e. unusable video quality), the DARA approach ensures that usable video quality is maintained for all sensors. Thus, DARA does not only lead to the highest average PSNR across all sensors, but it also ensures that the minimum PSNR achieved is enough to sustain usable video quality under the imposed delay constraints.

In Table VI, we report a summary of the average PSNR achieved by different policies in different scenarios. As mentioned previously, PSNR values above 20dB were deemed to correspond to usable visual quality for the spatial resolution and frame rate of the utilized video material.

\begin{table}
\caption{Average PSNR (in dB) of different policies.}\label{Table6}
\vspace{-10pt}
\centering{\includegraphics[scale = 0.8]{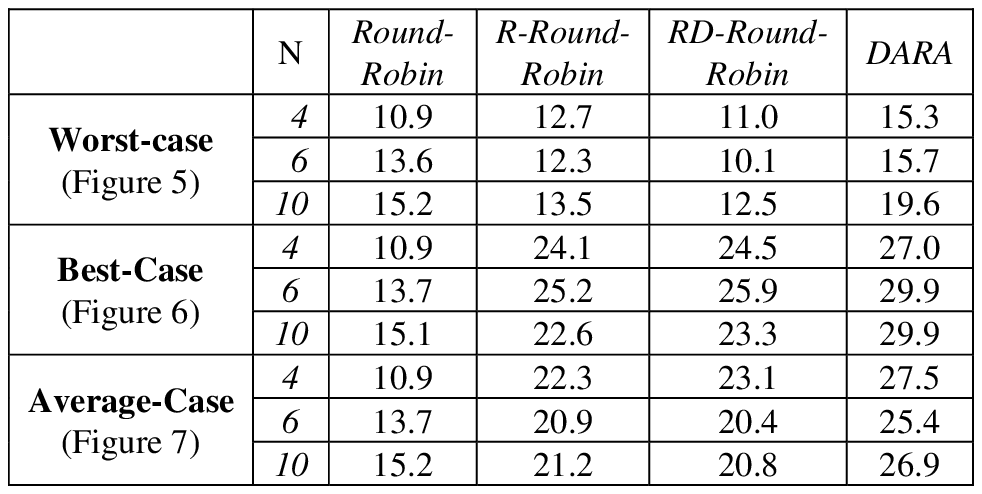}}
\vspace{-15pt}
\end{table}

\vspace{-10pt}
\section{Conclusions}
We present a new solution for slot allocation of multi-camera video streaming that is validated under IEEE 802.15.4e wireless personal area networks, which are expected to become a dominant deployment framework for video streaming under the Internet-of-Things (IoT) paradigm. Within the resource-constrained context of IoT applications, unlike existing works that require packet-level cross-layer information to perform resource allocation, we consider the more practical case where only limited statistical information of each sensor's packet deadlines is available. A unique characteristic of the proposed DARA approach for resource allocation is that it yields a non-stationary slot allocation policy by updating the indices in a manner that depends on the allocation of previous slots. This is in contrast with all existing slot allocation policies such as round-robin and its variants, which are stationary because the allocation of the current slot does not adapt to the allocation of previous slots. Our numerical studies and experiments with H.264/AVC encoded video and the IETF 6tisch simulator of the IEEE 802.15.4e TSCH indicate significant performance improvements against benchmark solutions. The present algorithm is constructed to operate in a particular setting but it can be applicable to many other resource allocation problems in many other settings.


%

%
%
%
%
%

\ifCLASSOPTIONcaptionsoff
  \newpage
\fi



%

%

%
\begin{IEEEbiographynophoto}{Jie Xu}
received the B.S. and M.S. degrees in Electronic Engineering from Tsinghua University, Beijing, China, in 2008 and 2010, respectively. He is currently a Ph.D. student with the Electrical Engineering Department, UCLA. His primary research interests include wireless communications, game theory, learning and resource allocation in multi-agent networks.
\end{IEEEbiographynophoto}

\begin{IEEEbiographynophoto}{Yiannis Andrepoulos}
obtained the Electrical Engineering Diploma and an MSc degree from the University of Patras, Patras, Greece. He obtained the PhD in Applied Sciences from the University of Brussels (Belgium) in May 2005. During his post-doctoral work at the University of California Los Angeles (US) he performed research on cross-layer optimization of wireless media systems, video streaming, and theoretical aspects of rate-distortion-complexity modeling for multimedia stream processing systems. From Oct. 2006-Dec. 2007, he was Lecturer at the Electronic Engineering Department of Queen Mary University of London. Since Dec. 2007, he is Lecturer at the Electronic and Electrical Engineering Department of UCL. During 2002-2004, Dr. Andreopoulos made several decisive contributions to the ISO/IEC JTC1/SC29/WG11 (Moving Picture Experts Group ¨C MPEG) committee in the early exploration on scalable video coding. He is working in the fields of multimedia stream processing and coding, error-tolerant computing, signal processing \& transform design, and wireless protocols for low-end systems (e.g. sensor networks).
\end{IEEEbiographynophoto}

\begin{IEEEbiographynophoto}{Yuanzhang Xiao}
received the B.E. and M.E. degree in Electrical Engineering from Tsinghua University, Beijing, China, in 2006 and 2009, respectively. He is currently pursuing the Ph.D. degree in the Electrical Engineering Department at the University of California, Los Angeles. His research interests include game theory, optimization, communication networks, and network economics.
\end{IEEEbiographynophoto}

\begin{IEEEbiographynophoto}{Mihaela van der Schaar}
is Chancellor's Professor of Electrical Engineering at University of California, Los Angeles. She is the Editor in Chief of IEEE Transactions on Multimedia, a Distinguished Lecturer of the Communications Society for 2011-2012,  and was a member of the Editorial Board of the IEEE Journal on Selected Topics in Signal Processing. She received an NSF CAREER Award, Okawa Foundation Award (2006), IBM Faculty Award (2005, 2007, 2008), and several best paper awards. She has 33 granted US patents. She is also the founding director of the UCLA Center for Engineering Economics, Learning, and Networks.
\end{IEEEbiographynophoto}

%
%




\end{document}